\documentclass[11pt]{article}
\usepackage{fullpage}

\usepackage[letterpaper, margin=1in]{geometry}

\usepackage{authblk}
\usepackage{amsthm}
\usepackage{amsmath}
\usepackage[table,xcdraw]{xcolor}
\usepackage{multirow}
\usepackage{wrapfig}
\usepackage[leftcaption]{sidecap}
\usepackage{multicol}
\usepackage{thmtools}
\usepackage{thm-restate}
\usepackage{nicefrac}
\usepackage[boxed,lined,linesnumbered, ruled, noend]{algorithm2e}
\usepackage{comment,amsfonts,amsmath,amsthm,graphicx,subcaption}
\usepackage{wrapfig}

\usepackage{booktabs}

\usepackage{natbib}

\definecolor{Darkblue}{rgb}{0,0,0.4}
\definecolor{Brown}{cmyk}{0,0.61,1.,0.60}
\definecolor{Purple}{cmyk}{0.45,0.86,0,0}
\definecolor{Darkgreen}{rgb}{0.133,0.543,0.133}

\usepackage[colorlinks,linkcolor=Darkblue,filecolor=blue,citecolor=blue,urlcolor=Darkblue,pagebackref]{hyperref}
\usepackage[nameinlink]{cleveref}

\usepackage{enumitem}

\newcommand{\commentout}[1]{}

\def\epsilon{\varepsilon}
\def\eps{\varepsilon}

\newcommand{\alert}[1]{\textbf{\color{red}
		[[[#1]]]}\marginpar{\textbf{\color{red}**}}\typeout{ALERT:\@
		\the\inputlineno: #1}}

\usepackage[colorinlistoftodos,prependcaption,textsize=tiny]{todonotes}
\newcommand{\atodo}[1]{}
\newcommand{\atodoin}[1]{}

\newcommand{\gtodo}[1]{}
\newcommand{\gtodoin}[1]{}

\def\eps{\epsilon}

\newcommand{\cost}{{\rm cost}}

\newcommand{\opt}{{\rm OPT}}

\newcommand{\poly}{{\rm poly}}

\newcommand{\rad}{{\rm rad}}

\newcommand{\mommit}[1]{}
\newcommand{\namedref}[2]{\hyperref[#2]{#1~\ref*{#2}}}

\newcommand{\ball}{{\sf ball}}
\newcommand{\ballp}{\ensuremath{{\sf \cB^+}}}

 \SetAlgoNoEnd

\theoremstyle{plain}
\newtheorem{theorem}{Theorem}[]
\newtheorem{lemma}{Lemma}[]
\newtheorem{claim}[lemma]{Claim}

\newtheorem{definition}{Definition}
\newtheorem{corollary}[theorem]{Corollary}

\newtheorem{observation}{Observation}

\newtheorem{question}{Question}[]

\newcommand{\lsum}{{\sf LSUM}}
\newcommand{\rsum}{{\sf RSUM}}
\newcommand{\msum}{{\sf MSUM}}
\newcommand{\lmsum}{{\sf LMSUM}}

\newcommand{\Ext}{\textsf{Ex}}

\newcommand{\cA}{\mathcal{A}}
\newcommand{\cB}{\mathcal{B}}
\newcommand{\cC}{\mathcal{C}}
\newcommand{\cD}{\mathcal{D}}

\newcommand{\cF}{\mathcal{F}}

\newcommand{\cI}{\mathcal{I}}
\newcommand{\cJ}{\mathcal{J}}

\newcommand{\cL}{\mathcal{L}}

\newcommand{\cR}{\mathcal{R}}
\newcommand{\cS}{\mathcal{S}}

\newcommand{\cU}{\mathcal{U}}

\newcommand{\sor}{\textsf{Sum of Radii}\xspace}

\newcommand{\gbctime}{$2^{O(k\log k)} n$}

\newcommand{\np}{\textsf{NP}\xspace}
\newcommand{\fpt}{\textsf{FPT}}

\newcommand{\eth}{\textsf{ETH}}

\newcommand{\nonblind}[1]{}

\newcommand{\Wtwo}{\ensuremath{\textsc{W}[2]}\xspace}
\newcommand{\scover}{Set Cover\xspace}
\newcommand{\domset}{Dominating Set\xspace}

\newcommand{\msr}{$k$-MSR\xspace}

\newcommand{\abfair}{$(\alpha,\beta)$-fair \msr}
\newcommand{\abfairsc}{$(\alpha,\beta)$-Fair Clustering\xspace}

\newcommand{\met}{\ensuremath{(X,d)}\xspace}

\newcommand{\msrins}{\ensuremath{(\met,k)}}

\author{Ameet Gadekar}

\affil{CISPA Helmholtz Center for Information Security, Saarbr\"{u}cken, Germany.\qquad\qquad\qquad Email:  \texttt{ameet.gadekar@cispa.de}}

\title{On the Parameterized Approximability of\\ (Mergeable) Sum of Radii Clustering
}
\date{}
\begin{document}
\pagenumbering{gobble}

\maketitle

\begin{abstract}
The sum of radii problem ($k$-MSR) asks, given a metric space on $n$ points, to place $k$ balls covering all points so as to minimize the sum of their radii. Despite extensive study from the perspectives of approximation and parameterized algorithms, the exact parameterized complexity of the problem and the existence of efficient parameterized approximation schemes remained open. We advance this understanding on both the hardness and algorithmic fronts.

We begin by showing that $k$-MSR is $\textsc{W}[2]$-hard parameterized by $k$, thereby pinpointing its location in the $\textsc{W}$-hierarchy. Moreover, via our reduction, we rule out efficient parameterized approximation schemes (EPAS)—that is, $(1+\epsilon)$-approximations running in time $f(k,\epsilon)\cdot \mathrm{poly}(n)$—unless $\textsc{W}[2] = \textsc{FPT}$. Assuming the Exponential Time Hypothesis, we further rule out such algorithms running in time $f(k,\epsilon)\cdot n^{o(k)}$, strengthening recent lower bounds for the problem.

On the algorithmic side, we study $k$-MSR under the framework of \emph{mergeable constraints}, which captures a broad class of clustering constraints, including fairness, diversity, and lower bounds. We obtain an $\textsc{FPT}$ $(\frac{8}{3}+\epsilon)$-approximation, improving upon the previous best guarantee of $(4+\epsilon)$. Moreover, given access to a suitable assignment subroutine, we achieve a $(2+\epsilon)$-approximation, matching the best known bound for the unconstrained problem. This, in turn, yields $(2+\epsilon)$ $\textsc{FPT}$-approximations for several important settings, including  $(t,k)$-fair, $(\alpha,\beta)$-fair, $\ell$-diversity, and private clustering.
\end{abstract}



	\newpage

\pagenumbering{arabic}


\section{Introduction}

Clustering with the sum of radii objective has emerged as a natural alternative to classical center-based formulations such as $k$-median, $k$-center, and $k$-means. The unconstrained sum of radii problem, abbreviated as \msr, admits a simple geometric interpretation: given a metric space $(X,d)$ on $n$ points, the goal is to place $k$ balls that cover all points while minimizing the sum of their radii. A particularly appealing feature of \msr is that it mitigates the \emph{dissection effect} on underlying ground-truth clusters—a phenomenon that commonly occurs in its closely related cousin, the $k$-center problem, which instead aims to minimize the maximum radius.

From a computational perspective, the approximability landscape of $k$-center is well understood: it admits a simple $2$-approximation, and this bound is known to be tight~\cite{GONZALEZ1985293,hochbaum1985best}. In contrast, despite more than two decades of study~\cite{CHARIKAR2004417,GibsonKKPV12,BehsazS15}, the polynomial-time approximability of \msr remains far less understood. The problem is known to be \np-hard and admits a QPTAS~\cite{GibsonKKPV12}, thereby suggesting the possibility of a PTAS. Recently, \cite{BuchemERW24} improved the best known polynomial-time approximation factor to $(3+\eps)$, improving upon earlier results~\cite{friggstad_et_al,CHARIKAR2004417}.

Over the past two decades, \fpt-approximation algorithms have achieved remarkable success for clustering problems, obtaining approximation factors significantly better than those achievable in polynomial time~\cite{badoiu-etal:approximate-clustering-coresets, kumar2010linear, cohenaddad_et_al:LIPIcs.ICALP.2019.42,
Cohen-AddadL19, AbbasiFOCS23}. These algorithms run in time $f(k)\poly(n)$, and since the parameter $k$ is often small compared to $n$, they are particularly appealing in practice. Moreover, this framework has proved effective even in constrained settings, including capacities, fairness, diversity, and outliers; see, for example,~\cite{GOYAL2023190,goyal2023tight,Cohen-AddadL19,
bhore2026clusteringconstraintsefficientparameterized,DBLP:conf/kdd/ThejaswiGOO22,DBLP:journals/corr/GT26, DBLP:journals/corr/DLP26} and references therein.

For $k$-center, the parameterized approximability landscape is again well understood: no $(2-\eps)$-approximation algorithm exists in \fpt\ time unless $\Wtwo = \fpt$. In sharp contrast, the parameterized complexity of unconstrained \msr remains poorly understood. On the algorithmic side, the best known \fpt\ approximation guarantee is $(2+\eps)$, due to Chen et al.~\cite{ChenXXZ24}, based on an extension of the classical greedy framework of~\cite{hochbaum1985best} for $k$-center. On the lower-bound side, Banerjee et al.~\cite{DBLP:conf/aaai/BanerjeeBGH25} ruled out $n^{o(k)}$-time algorithms under the exponential time hypothesis (\eth), although the precise parameterized hardness class of the problem remained unresolved.

At the same time, under standard complexity assumptions, \msr is known to admit a PTAS, which in turn implies the existence of a parameterized approximation scheme\footnote{Algorithms that run in time $f(k,\eps)n^{g(\eps)}$ and output a $(1+\eps)$-approximation for every $\eps>0$.} (PAS), even though no explicit algorithm is currently known. Thus, from the perspective of parameterized approximability, the main remaining question is whether \msr admits an \emph{efficient} parameterized approximation scheme (EPAS), namely a PAS whose running time is of the form $f(k,\eps)\poly(n)$.

\begin{question}\label{ques:one}
Does \msr admit an EPAS?
\end{question}

More broadly, while recent lower bounds establish parameterized intractability under \eth, the precise parameterized hardness class of the problem remained unknown.

Our first result resolves these questions. We show that \msr is \Wtwo-hard, thereby pinpointing its location in the \textsc{W}-hierarchy.\footnote{The decision version of \msr is in \Wtwo via a reduction to partitioned \scover: for radii $r_1,\dots,r_k$, construct $k$ families of balls of radius $r_i$ and select one from each family to cover all points.} In particular, this rules out exact \fpt\ algorithms under the standard assumption that $\Wtwo \neq \fpt$, which is strictly weaker than \eth. Moreover, we rule out EPAS, thereby answering \Cref{ques:one} in the negative and essentially completing the parameterized approximability landscape of \msr.

\begin{restatable}{theorem}{MSRHardness}\label{thm:hardness}
The \msr problem is \Wtwo-hard when parameterized by $k$. Furthermore:
\begin{itemize}
    \item (No EPAS) Unless $\Wtwo = \fpt$, there is no algorithm with running time $f(k,\eps)\cdot \poly(n)$ that computes a $(1+\eps)$-approximation for every $\eps > 0$.
    \item Assuming \eth, there is no algorithm with running time $f(k,\eps)\cdot n^{o(k)}$ that computes a $(1+\eps)$-approximation for every $\eps > 0$.
\end{itemize}
\end{restatable}

To prove~\Cref{thm:hardness}, we reduce (partitioned) \scover to \msr such that, if there exists a set cover of size $k$, then the constructed \msr instance admits a solution of cost at most $2^k-1$, whereas if no such set cover exists, then every solution has cost at least $2^k$. Since \scover is \Wtwo-hard and the reduction preserves the parameter, this implies that \msr is \Wtwo-hard. Moreover, the gap between the yes and no instances is sufficient to rule out an EPAS under the assumption that $\Wtwo \neq \fpt$.

We next turn to the design of \fpt\ approximation algorithms for \msr. As discussed earlier, the \fpt\ framework has been highly successful not only in the unconstrained setting but also in a wide range of constrained variants. For \msr, recent work has developed \fpt\ approximation algorithms under capacities~\cite{DBLP:journals/corr/FG25,jaiswal_et_al:LIPIcs.ITCS.2024.65}, fairness constraints~\cite{ChenXXZ24,DBLP:conf/isaac/CartaDHR024,DBLP:conf/approx/BandyapadhyayC25,DBLP:conf/aaai/BanerjeeBGH25}, and outliers~\cite{DBLP:journals/corr/LGXZ25,gadekar2026fptapproximationsfairsum}.

While these constraints are well motivated by real-world applications, they introduce significant challenges in enforcing feasibility. More importantly, the rapidly growing variety of such constraints makes it increasingly difficult to design algorithms tailored to each individual setting. For instance, there are already more than a dozen notions of fairness constraints (see~\cite{FairClustering} and references therein), and new variants continue to emerge.

\paragraph*{Mergeable constraints.}
To address this challenge, Arutyunova and Schmidt~\cite{DBLP:conf/stacs/Arutyunova021} introduced a clean and general framework of \emph{mergeable constraints}. Informally, a constraint is said to be mergeable if merging any two clusters of a feasible solution results in another feasible (possibly larger) cluster. This abstraction captures a broad range of commonly studied constraints~\cite{DBLP:conf/isaac/CartaDHR024}, including the following:
\begin{itemize}[itemsep=1pt, topsep=1pt]
    \item \emph{Balanced Clustering}~\cite{DBLP:journals/corr/BFLS20}: Points are colored either red or blue, and each cluster must contain equal numbers of red and blue points.
    
    \item \emph{Pairwise Fair Clustering}~\cite{DBLP:conf/nips/Chierichetti0LV17}: Points are colored red or blue, and each cluster must maintain a bounded ratio between the number of blue and red points, lying within $[1/t, t]$. This is also known as $(t,k)$-fair clustering.
    
    \item \emph{Pairwise Exact with Multiple Colors}~\cite{DBLP:conf/icalp/Rosner018}: Points  have multiple colors, and each cluster must preserve the same proportion of each color.
    
    \item \emph{$\ell$-Diversity Clustering}~\cite{DBLP:journals/algorithmica/DingX20, DBLP:journals/jcss/BandyapadhyayFS24}: Points  have multiple colors, and no color  constitute more than a $1/\ell$ fraction of any cluster.
    
    \item \emph{Fair Representational}~\cite{DBLP:conf/nips/BeraCFN19,DBLP:conf/approx/Bercea0KKRS019}: Points  have multiple colors, and each color must lie within specified lower and upper bounds in every cluster.
    
    \item \emph{$(\alpha,\beta)$-Fair Clustering}~\cite{DBLP:conf/approx/Bercea0KKRS019, DBLP:journals/jcss/BandyapadhyayFS24}: A generalization of fair representational clustering in which a point may have multiple colors simultaneously. It captures all of the above constraints.
    
    \item \emph{Private Clustering}~\cite{DBLP:conf/icalp/Rosner018}: Each cluster must contain at least $L$ points; this is also known as clustering with a uniform lower bound.
\end{itemize}
\medskip

In their work, Carta et al.~\cite{DBLP:conf/isaac/CartaDHR024} studied mergeable constraints for \msr and designed a $(6+\eps)$ \fpt-approximation for \msr under mergeable constraints, yielding a unified guarantee across all these settings. They also obtained improved $(3+\eps)$ approximations for special cases such as balanced clustering and uniform lower bounds. Subsequently, Bandyapadhyay and Chen~\cite{DBLP:conf/approx/BandyapadhyayC25} improved the general factor to $(4+\eps)$, which is the current best known bound.

Our second main result significantly improves this factor to $(\tfrac{8}{3}+\eps)$.

\begin{restatable}{theorem}{MainApx}\label{thm:mainapx}
There is an algorithm for \msr under a mergeable constraint that, for every $\eps>0$, computes a $(\tfrac{8}{3}+\eps)$-approximation in time $2^{O(k \log (k/\eps))} \cdot \poly(n)$.
\end{restatable}

At a high level, our algorithm follows the blueprint of Bandyapadhyay and Chen~\cite{DBLP:conf/approx/BandyapadhyayC25}, but introduces key new ingredients. Their approach first computes an unconstrained clustering, then expands the clusters, and finally merges them to obtain feasibility. In other words, it constructs an unconstrained solution and subsequently \emph{projects} it onto a feasible one, incurring losses in both steps.
In contrast, our algorithm operates entirely within the mergeable-constraints framework, thereby avoiding this projection step altogether. This structural difference, together with a more refined charging scheme and a careful analysis, leads to an improved approximation factor.

Our final result further improves the approximation to $(2+\eps)$, assuming access to an assignment subroutine for the underlying constraint. Such a subroutine takes a set of at most $k$ balls and assigns each point to the center of a ball containing it, while respecting the constraint, provided such an assignment exists.

\begin{restatable}{theorem}{TwoApx}\label{thm:2apx}
Suppose there exists an assignment subroutine $\cA_A$ with running time $T(\cA_A)$. 
Then there is an algorithm that computes a $(2+\eps)$-approximation for \msr under a mergeable constraint in time $2^{O(k \log (k/\eps))} \cdot n^2 \cdot T(\cA_A)$.
\end{restatable}

This result matches the best known \fpt\ approximation factor for unconstrained \msr~\cite{ChenXXZ24}, where the assignment step is trivial. Moreover, it improves the best known guarantees for \abfair, which captures a broad class of fair clustering problems, improving upon the $(3+\eps)$ bound of~\cite{DBLP:conf/isaac/CartaDHR024} for several settings. To this end, we design an \fpt-time assignment subroutine for $(\alpha,\beta)$-fair clustering based on ideas from~\cite{DBLP:journals/jcss/BandyapadhyayFS24}.

\begin{corollary}\label{cor:main}
There is an algorithm for \abfair that, for every $\eps>0$, computes a $(2+\eps)$-approximation in time $2^{O(k \log (k/\eps))} \cdot \poly(n)$. In particular, this yields $(2+\eps)$-approximation algorithms for \msr under (i) fair-representational, (ii) $\ell$-diversity, (iii) pairwise exact with multiple colors, and (iv) balanced constraints.
\end{corollary}

Using an assignment subroutine from~\cite{DBLP:conf/isaac/CartaDHR024}, we also obtained improved approximation factor for private clustering, improving the previous factor of $(3+\eps)$~\cite{DBLP:conf/isaac/CartaDHR024}.
\begin{corollary}\label{cor:second}
    There exists an FPT $(2+\eps)$-approximation algorithm for private \msr running in time $2^{O(k \log (k/\eps))} \cdot \poly(n)$.
\end{corollary}

Our algorithmic results, along with the current best approximation factors, are summarized in~\Cref{tab:tab1}.

\begin{table}[t]
\centering
\renewcommand{\arraystretch}{1.35}
\setlength{\tabcolsep}{10pt}
\begin{tabular}{lccc}
\toprule
\textbf{Setting} & \textbf{Best Known} & \textbf{Our Result} & \textbf{Reference} \\
\midrule
\textit{Mergeable constraints} 
& $(4+\eps)$~\cite{DBLP:conf/approx/BandyapadhyayC25} 
& \textbf{$(\tfrac{8}{3}+\eps)$} 
& \Cref{thm:mainapx} \\

\textit{$(\alpha,\beta)$-fair} 
& $(4+\eps)$~\cite{DBLP:conf/approx/BandyapadhyayC25} 
& \textbf{$(2+\eps)$} 
& \Cref{thm:2apx} \\

\textit{Balanced} 
& $(3+\eps)$~\cite{DBLP:conf/isaac/CartaDHR024} 
& \textbf{$(2+\eps)$} 
& \Cref{cor:main} \\

\textit{Private} 
& $(3+\eps)$~\cite{DBLP:conf/isaac/CartaDHR024} 
& \textbf{$(2+\eps)$} 
& \Cref{cor:second} \\
\bottomrule
\end{tabular}
\caption{Improved \fpt\ approximation guarantees for \msr under mergeable constraints. }
\label{tab:tab1}
\end{table}







\paragraph*{Hardness implications and limitations of prior reductions.}
Both of our hardness results— \Wtwo-hardness and the non-existence of an EPAS—extend to the mergeable-constraint setting, since the unconstrained version trivially satisfies the mergeability property. A result of~\cite{DBLP:conf/approx/BandyapadhyayC25} claims to rule out algorithms with running time $f(k)\,n^{o(k)}$ for fair-representational \msr, which would in turn imply such hardness for the broader mergeable-constraint setting. 
However, we observe that the reduction in~\cite{DBLP:conf/approx/BandyapadhyayC25} contains a subtle flaw that invalidates the claimed hardness. In this light, our result not only restores the intended hardness conclusion but also strengthens it by ruling out an EPAS. See~\Cref{ss:hardissue} for more details.

\subsection{Further Related Work}

\paragraph*{Polynomial-time regime.}
Behsaz and Salavatipour~\cite{BehsazS15} showed that \msr can be solved exactly in polynomial time on unweighted graphs, provided that zero-radius clusters are disallowed. In Euclidean metrics of constant dimension, \cite{GibsonKKPV12} gave a polynomial-time exact algorithm, and showed that the problem is \np-hard in metrics of constant doubling dimension. They further established \np-hardness for shortest-path metrics induced by weighted planar graphs, while also presenting a QPTAS for general metrics. Moreover, even when the locations of the optimal centers are known, computing a valid assignment of points to centers remains \np-hard~\cite{GibsonKKPV12,DBLP:journals/corr/DHKS23}. For \msr under pairwise fair constraints, \cite{DBLP:conf/esa/NezhadBC25} recently obtained an $O(1)$-approximation in polynomial time.

\paragraph*{\fpt\ regime.} 
For (unconstrained) \msr, \cite{DBLP:conf/compgeom/BandyapadhyayL023a} showed an EPAS in Euclidean metrics. When parameterizing by the treewidth of the underlying metric graph, \cite{DBLP:conf/waoa/DrexlerHLSW23} obtained an \textsf{XP}-time algorithm. 

Chen et al.~\cite{ChenXXZ24} obtained a $(2+\eps)$-approximation for \msr in general metrics. Their approach can be viewed as an extension of the classical algorithm of~\cite{hochbaum1985best} for $k$-center: repeatedly selecting a point among the remaining points, guessing the radius of its optimal cluster, and removing all points within twice that radius. They also extend this approach to fair and matroid constraints, incurring an additional loss in the approximation factor.

Under capacity constraints, the best known \fpt-approximation factor is due to~\cite{DBLP:journals/corr/FG25}. Liu et al.~\cite{DBLP:journals/corr/LGXZ25} studied an extension of \msr with outliers, called colorful \msr, which imposes separate budgets on the number of outliers for each color class, and obtained a $(7+\eps)$ \fpt-approximation. Recently, \cite{gadekar2026fptapproximationsfairsum} considered \msr under simultaneous fairness (on centers) and outlier constraints, and presented a $(3+\eps)$ \fpt-approximation that extends to objectives defined by monotone symmetric norms of radii, beyond the sum-of-radii objective.

For mergeable constraints, \cite{DBLP:conf/waoa/DrexlerHLSW23} designed an EPAS for \msr in Euclidean metrics of arbitrary dimension. This was subsequently extended to metrics of constant doubling dimension by \cite{DBLP:conf/aaai/BanerjeeBGH25}.

\paragraph*{Organization.}
We begin with preliminaries and problem definitions in~\Cref{sec:prelim}. In~\Cref{sec:algo}, we present our algorithms, including a $(2+\eps)$-approximation in~\Cref{ss:2apx}, a simple $(4+\eps)$-approximation in~\Cref{ss:4apx} matching the previous best known factor, and our main $(\tfrac{8}{3}+\eps)$-approximation algorithm in~\Cref{ss:mainsec}. In~\Cref{ss:assgn}, we present our assignment subroutines. Finally, in~\Cref{sec:hard}, we prove our hardness results.

\section{Preliminaries}\label{sec:prelim}
\paragraph*{Notations.}
Given a metric space $(X,d)$, for $x \in X$ and $r \in \mathbb{R}_{\geq 0}$, we denote by $\ball(x,r)$ as the ball centered at $x$ with radius $r$, i.e., $\ball(x,r):= \{p\in X: d(x,p)\le r\}$. 
For a ball $B=\ball(p,R)$ and a radius $r_t$, let $\Ext(B,r_t)$ denote  $\ball(p,R+r_t)$.

We now define the problems of our interest.
\begin{definition}[\sor]
   Given a metric space \met on $n$ points and a non-negative integer $k$, the task is to find a set $C \subseteq X$ of size at most $k$ and an assignment function $\sigma: X \rightarrow C$ such that $\sum_{c \in C} \rad(c)$ is minimized, where for $c\in C$,  $\rad(c):=\max_{p \in \sigma^{-1}(c)} d(p,c)$, is the radius of cluster centered at $c$.
\end{definition}
For brevity, we use \msr to mean \sor.
We denote by $\cI=\msrins$ an instance of \msr.
We call the pair $(C,\sigma)$ as a \emph{solution}  to $\cI$ or \emph{clustering} of $\cI$, and $\sigma^{-1}(c)$ as the \emph{cluster} centered at $c$. Furthermore, we define cost of the solution as $\cost_{\cI}((C,\sigma)) = \sum_{c \in C}\rad(c)$, the sum of radii of the clusters given by $\sigma^{-1}$. Given a clustering $(C,\sigma)$ of $\cI$, for $c_i,c_j \in C$, we \emph{merge} two clusters $\Pi_i:= \sigma^{-1}(c_i)$ and $\Pi_j: \sigma^{-1}(c_j)$ by setting $\sigma(p)=c$ for all points $p \in \Pi_i \cup \Pi_j$ and some $c \in X$. Finally, a solution $(O,\sigma^*)$ to $\cI$ with minimum cost is an \emph{optimal} solution or clustering, and denote its cost by $\opt(\cI)$ or $\opt$, when $\cI$ is clear from the context. Additionally, let $R^*$ denote the set of radii corresponding to the optimal solution $(O,\sigma^*)$.

As mentioned in the introduction, \msr has been extensively studied under several constraints, including capacities, outliers, and fairness. We say a clustering is \emph{feasible} w.r.t. some constraint $\cD$ if it satisfies all the constraints specified by $\cD$.
We say a constraint $\cD$ is \emph{mergeable} if for any feasible clustering w.r.t. $\cD$, if  merging any two of its clusters results in a clustering that is feasible w.r.t. $\cD$.
We consider the \msr problem under mergeable constraints.
Similar to prior works~\cite{DBLP:conf/isaac/CartaDHR024,DBLP:conf/approx/BandyapadhyayC25}, we assume that, given a candidate clustering, one can efficiently verify whether it satisfies the mergeable constraint.
Next, we define assignment subroutine for \msr under mergeable constraints.

\begin{definition}[Assignment subroutine]\label{def:assgn}
Let $\cI$ be an instance of \msr under a mergeable constraint. 
An algorithm $\cA_A$ is called an \emph{assignment subroutine} for $\cI$ if, given a set of balls 
$\cB=\{B_1=\ball(c_1,R_1),\dots, B_k=\ball(c_k,R_k)\}$, it either

\begin{itemize}
    \item returns a feasible clustering $(C=\{c_1,\dots,c_k\},\sigma)$ such that 
    $\sigma(p) \in \{c_i \in C : p \in B_i\}$ for all $p \in X$, or
    \item correctly reports that no such feasible clustering exists.
\end{itemize}
\end{definition}
Note that any clustering returned by an assignment subroutine on $\cB$ has cost at most $\cost(\cB)$.

A \emph{radius profile} is a vector $R=(r_1,\dots,r_\ell)$ of non-negative reals, 
where $\ell \le k$. For a given instance $\cI$ of (constrained) \msr, 
the radius profile of a clustering $(C,\sigma)$ is the vector of the radii 
of its clusters, under some fixed ordering of the clusters.

We say that $R^*$ is an \emph{optimal radius profile} for $\cI$ if there exists 
an optimal solution to $\cI$ whose radius profile is $R^*$.

We say that a radius profile $R=(r_1,\dots,r_\ell)$ is an 
\emph{$\eps$-approximation} of an optimal radius profile 
$R^*=(r^*_1,\dots,r^*_\ell)$ if
$
r^*_i \le r_i \quad \text{for all } i \in [\ell], 
 \text{and} 
\sum_{i=1}^\ell r_i \le (1+\eps)\sum_{i=1}^\ell r^*_i.
$
We assume that $R$ and $R^*$ have the same length $\ell$, corresponding to the 
same number of clusters under a fixed ordering.

It is known that for any instance of (constrained) \msr, one can compute a list 
of radius profiles that contains an $\eps$-approximation of every optimal radius profile.

\begin{theorem}[\cite{jaiswal_et_al:LIPIcs.ITCS.2024.65,DBLP:journals/corr/FG25}]\label{thm:epsapxrad}
There is an algorithm that, for any instance of (constrained) \msr and any $\eps>0$, 
computes a list $\cL$ of radius profiles of size $2^{O(k \log (k/\eps))} n^2$ such that 
$\cL$ contains an $\eps$-approximation of every optimal radius profile for $\cI$. 
Moreover, the algorithm runs in time $2^{O(k \log (k/\eps))} n^2$.
\end{theorem}


     

\section{Improved FPT Approximations}\label{sec:algo}
In this section, we present our \fpt-approximation algorithms for \msr under mergeable constraints. First, we setup the required definitions.

\paragraph*{Definitions.}
For the following definitions, we fix an instance $\cI=\msrins$ of \msr under a mergeable constraint.
Let $\Pi^* = \{\Pi_1^*,\dots,\Pi_k^*\}$ be the clusters of a fixed optimal clustering for $\cI$.

\begin{definition}[Ball Cover]\label{def:ballcover}
A set $\cB$ of center-disjoint balls  is a \emph{$k$-ball cover} for $\cI$ if $|\cB|\le k$ and $X \subseteq \bigcup_{B \in \cB} B$.
\end{definition}

\begin{definition}[$\Pi^*$-Respecting Cover]\label{def:pirespect}
A set $\cB$ of balls is a \emph{$\Pi^*$-respecting cover} if:
\begin{itemize}
    \item $\cB$ is a $k$-ball cover for $(X,d)$, and
    \item for every ball $B=\ball(p,R) \in \cB$, if $p \in \Pi_i^*$ for some $i \in [k]$, then
    $\Pi_i^* \subseteq B$.
\end{itemize}
\end{definition}

\begin{definition}[$\Pi^*$-Feasible Cover]\label{def:pifeas}
A set $\cB$ of balls is a \emph{$\Pi^*$-feasible cover} if:
\begin{itemize}
    \item $\cB$ is a $\Pi^*$-respecting cover, and
    \item for every $i \in [k]$, there exists a ball $B \in \cB$ such that
    $\Pi_i^* \subseteq B$.
\end{itemize}
A set $\cB$ of balls is a \emph{feasible cover} if there exists an optimal clustering $\Pi^*$ such that $\cB$ is a $\Pi^*$-feasible cover.
\end{definition}


\begin{figure}[t]
    \centering
    
    \begin{subfigure}{0.2\textwidth}
        \centering
        \includegraphics[width=\linewidth]{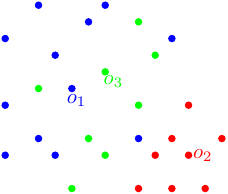}
    \end{subfigure}
    \begin{subfigure}{0.38\textwidth}
        \centering
        \includegraphics[width=\linewidth]{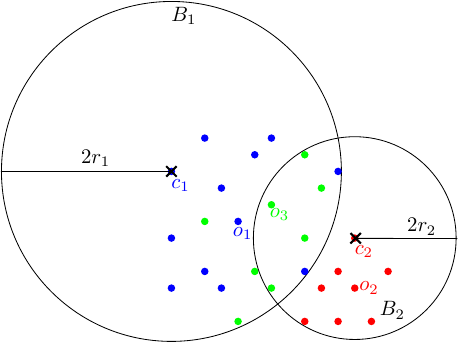}
    \end{subfigure}
    \begin{subfigure}{0.38\textwidth}
        \centering
        \includegraphics[width=\linewidth]{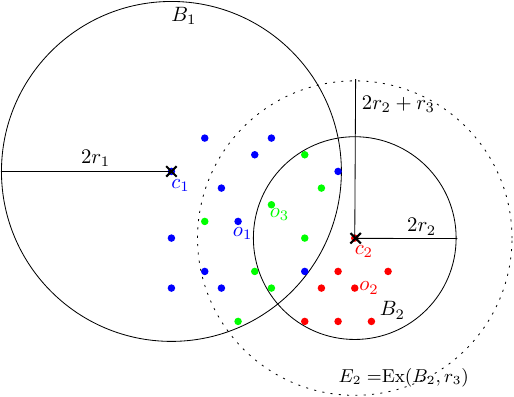}
    \end{subfigure}

\caption{\small 
The left figure shows an optimal clustering $\Pi^*$ with centers $o_1,o_2,o_3$ and radii $r_1,r_2,r_3$, where points of the same color belong to the same cluster. The middle figure shows a $\Pi^*$-respecting cover (obtained from GBC) consisting of  $B_1=\ball(c_1,2r_1)$ and $B_2=\ball(c_2,2r_2)$, with total cost $2(r_1+r_2)$. The right figure shows a $\Pi^*$-feasible cover $\{B_1,E_2\}$ (obtained from FC using $B_1,B_2$) with total cost $2r_1+2r_2+r_3 \le 2\opt$.
}
    \label{fig:covers}
\end{figure}
See~\Cref{fig:covers} for an illustration of these covers.
For a set of balls $\cB$, let $\cost(\cB)$ denote the sum of the radii of its balls. We begin with the following simple observation.

\begin{observation}\label{obs:fc}
    Let $\cF$ be a feasible cover for an instance $\cI$ of \msr under a mergeable constraint. Then there exists a feasible solution $(C,\sigma)$ to $\cI$ with cost at most $\cost(\cF)$, where $C$ is the set of centers of the balls in $\cF$.
\end{observation}

\begin{proof}
    Let $\Pi^*$ be an optimal clustering of $\cI$ such that $\cF$ is a $\Pi^*$-feasible cover. For each cluster $\Pi^*_i \in \Pi^*$, fix a ball $B_{j_i} \in \cF$ satisfying $\Pi^*_i \subseteq B_{j_i}$. Now, for every cluster $\Pi^*_i$, assign all its points  to the center of $B_{j_i}$.
    Consequently, each cluster in the resulting solution is a union of clusters from $\Pi^*$. Since the constraint is mergeable, all resulting clusters are feasible. Finally, the cost of the solution is at most $\cost(\cF)$, as desired.
\end{proof}

\paragraph*{Overview.}
Our first goal is to compute a low-cost feasible cover for an instance of \msr under a mergeable constraint. Towards this, in Section~\ref{ss:gbc}, we design an \fpt\ algorithm (GBC) that computes a ball cover $\cB$ that is $\Pi^*$-respecting, with cost at most $(2+\eps)$ times the optimum for some optimal solution $\Pi^*$. Using $\cB$, we then present in Section~\ref{ss:fc} an \fpt\ algorithm (FC) that expands the balls in $\cB$ to obtain a $\Pi^*$-feasible cover $\cF$ that still has cost $(2+\eps)$ times the optimum.

Building on this feasible cover, we derive a $(2+\eps)$-approximation using an assignment subroutine and a $(4+\eps)$-approximation without such a subroutine, presented in Sections~\ref{ss:2apx} and~\ref{ss:4apx}, respectively. Finally, in Section~\ref{ss:mainsec}, we show how to leverage $\cF$ to obtain an $(\tfrac{8}{3}+\eps)$-approximation.
The derivation of our results and the interplay between the key theorems are illustrated in~\Cref{fig:algoflow}.




\begin{figure}
	\centering
	\includegraphics[width=0.85\linewidth]{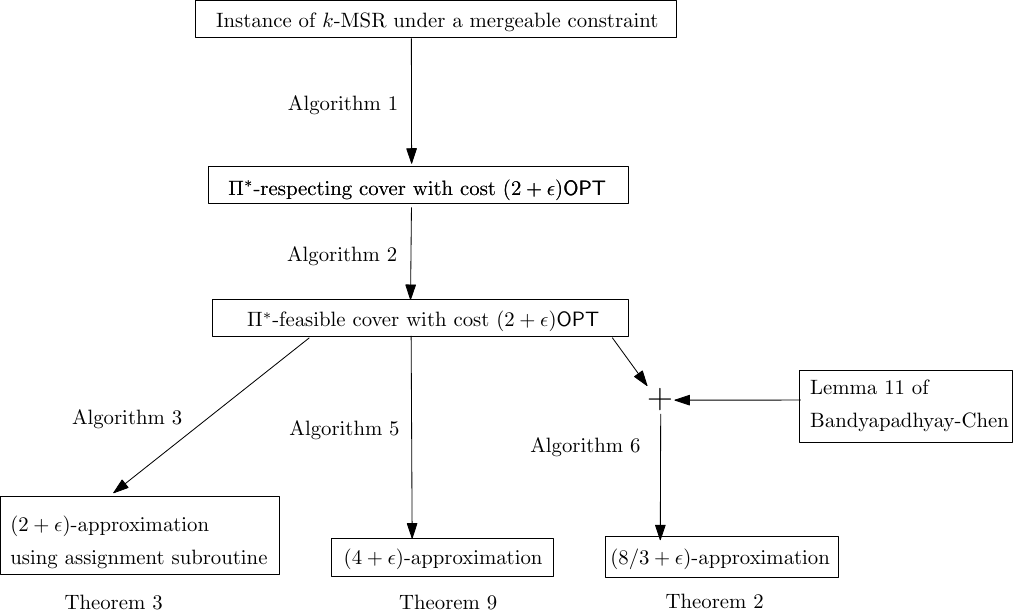}
	\caption{\small Derivation of our algorithmic results.
	}
	\label{fig:algoflow}
\end{figure}

\subsection{Greedy Ball Cover}\label{ss:gbc}
In this section, we present a simple greedy algorithm that, for any $\eps>0$, computes a ball cover for any given instance of \msr under mergeable constraints whose cost is at most $(2+\eps)$ times the optimal cost. The algorithm runs in time \gbctime.
\begin{theorem}[Greedy Ball Cover]\label{thm:gbc}
There exists an algorithm that, given an instance $\cI$ of \msr under mergeable constraints and a radius profile $R'$, 
runs in time \gbctime\ and returns a family $\mathbb{B}$ of $k$-ball covers such that the following holds:

For every $\eps > 0$ and every optimal clustering $\Pi^*$ for $\cI$ with radius profile $R^*$, 
if $R'$ is an $\eps$-approximation of $R^*$, then there exists a ball cover $\cB \in \mathbb{B}$ such that $\cB$ is a $\Pi^*$-respecting cover and
$\cost(\cB) \le (2+\eps)\opt(\cI)$.
\end{theorem}

\begin{proof}
Our algorithm is similar to the algorithms of~\cite{ChenXXZ24} and~\cite{DBLP:conf/approx/BandyapadhyayC25} for unconstrained \msr, which in turn is build on an extension of Gonzalez’s algorithm for $k$-center~\cite{gonzalez1985clustering}. The idea is to iteratively pick a point $p$ from the remaining points, guess the radius $r^*_p$ of the optimal cluster containing $p$, and remove all points in $\ball(p,2r^*_p)$. For correct guesses, each iteration removes all points from at least one optimal cluster, and hence the process terminates in at most $k$ iterations. Moreover, the selected points can serve as centers, and each such center $p$ can cover all points of its optimal cluster within distance $2r^*_p$, yielding a $2$-approximate solution for the unconstrained \msr.

We extend this idea to the constrained setting. Fix an optimal clustering $\Pi^* = \{\Pi_1^*,\dots,\Pi_k^*\}$ with centers $o_1^*,\dots,o_k^*$. Let $R^*$ denote the multiset of radii of the clusters in $\Pi^*$. Fix an arbitrary ordering of the clusters and represent $R^*$ as a vector $(r_1^*,\dots,r_k^*)$.\footnote{If the optimal solution uses fewer than $k$ clusters, we pad $R^*$ with zeros to obtain a vector of length $k$. This padding is purely notational and does not correspond to actual clusters.} and $R'=(r'_1,\dots,r'_k)$ be an $\eps$-approximation of $R^*$. Under this ordering, each cluster $\Pi_i^*$ is associated with an index $i \in [k]$.

We now describe how to use $R'$ within the greedy procedure. The pseudo-code is given in~\Cref{algo:gbc}.

\begin{algorithm}[h]
\caption{Greedy Ball Cover (GBC)} \label{algo:gbc}
\KwData{Instance $\cI=\msrins$ under mergeable constraint, and a radius profile $R'=(r'_1,\dots,r'_k)$}
\KwResult{A collection of at most $k$ balls}

$X' \leftarrow X$; 
$\cB \leftarrow \emptyset$\;

\While{$X' \neq \emptyset$ and $|\cB| < k$}
{   
    pick an arbitrary point $p \in X'$\;
    guess the index $i \in [k]$ of the optimal cluster containing $p$\;
    $X' \gets X' \setminus \ball(p,2r'_i)$\;
    $\cB \gets \cB \cup \{\ball(p,2r'_i)\}$\;
}

\Return $\cB$ if $\bigcup_{B \in \cB} B \supseteq X$, \textbf{else fail};
\end{algorithm}

If all guesses are correct, then in each iteration the algorithm removes all points from at least one cluster of $\Pi^*$. Hence, after at most $k$ iterations, we obtain $X' = \emptyset$, and the algorithm does not fail. Moreover, for each selected ball $B=\ball(p,2r'_i)$, the center $p$ lies in $\Pi_i^*$, and for every $x \in \Pi_i^*$,
\[
d(p,x) \le d(p,o_i^*) + d(o_i^*,x) \le r_i^* + r_i^* \le 2r'_i,
\]
where the last inequality follows since $r'_i \ge r_i^*$. Thus, $\Pi_i^* \subseteq B$, and hence $\cB$ is a $\Pi^*$-respecting cover.
The cost of $\cB$ is at most $2 \cdot \sum_{i=1}^k r'_i \le (2+\eps)\opt(\cI)$.

We eliminate guessing by branching over all possible choices of indices $i \in [k]$ in each iteration. Since there are at most $k$ iterations, this yields at most $k^k$ branches. For each branch, we obtain a collection of balls $\cB$ and add it to $\mathbb{B}$. Since there exists a branch corresponding to the correct sequence of guesses, at least one branch adds a $\Pi^*$-respecting cover.
\end{proof}

\subsection{From Respecting Cover to Feasible Cover}\label{ss:fc}

In this section, we show how to transform a ball cover into a feasible cover using an approximate radius profile. Since we do not know how to efficiently verify whether a given ball cover is a feasible cover, our algorithm instead returns a family of covers that is guaranteed to contain a feasible one.

\begin{theorem}\label{thm:feascover}
There exists an algorithm that, given an instance $\cI$ of \msr under mergeable constraints, 
a $k$-ball cover $\cB$ for $\cI$, and a radius profile $R'$, runs in time $2^{O(k \log k)} \cdot n$ 
and returns a family $\mathbb{F}$ of $k$-ball covers such that the following holds:

For every optimal clustering $\Pi^*$ for $\cI$ with radius profile $R^*$, 
if $\cB$ is a $\Pi^*$-respecting cover obtained by the greedy ball cover (GBC) procedure using $R'$, 
and $R'$ is an $\eps$-approximation of $R^*$, 
then there exists a cover $\cF \in \mathbb{F}$ such that $\cF$ is a $\Pi^*$-feasible cover and
\[
\cost(\cF) \le (2+\eps)\opt(\cI).
\]
\end{theorem}

\begin{algorithm}[h]
\caption{Feasible Cover Candidates (FC)} \label{algo:fc}
\KwData{Instance $\cI$, $k$-ball cover $\cB=\{B_1,\dots,B_m\}, m\le k$, radius profile $R'=(r'_1,\dots,r'_k)$}
\KwResult{A family $\mathbb{F}$ of $k$-ball covers}

$\mathbb{F} \leftarrow \emptyset$\;

\ForAll{$(t_1,\dots,t_m) \in (\{0\} \cup [k])^m$}{
    \ForEach{$B_\ell = \ball(p_\ell,R_\ell) \in \cB$}{
        \eIf{$t_\ell = 0$}{
            $B'_\ell \leftarrow B_\ell$\;
        }{
            $B'_\ell \leftarrow \Ext(B_\ell, r'_{t_\ell})$\;
        }
    }
    $\cF \leftarrow \{B'_1,\dots,B'_m\}$\;
    add $\cF$ to $\mathbb{F}$\;
}

\Return $\mathbb{F}$;
\end{algorithm}

\begin{proof}
The pseudocode of our algorithm is shown in~\Cref{algo:fc}.
Fix an optimal clustering $\Pi^* = \{\Pi_1^*,\dots,\Pi_k^*\}$ with centers $o_1^*,\dots,o_k^*$ and radius profile $R^*=(r_1^*,\dots,r_k^*)$. 
Assume that $\cB$ is a $\Pi^*$-respecting cover obtained by the greedy ball cover procedure using $R'$, and that $R'$ is an $\eps$-approximation of $R^*$.

For each cluster $\Pi^*_i$, we define the \emph{responsible} ball in $\cB$ as follows. If there is a ball $B_{j} \in \cB$ whose center lies in $\Pi^*_i$, then  $B_{j}$ is responsible for $\Pi^*_i$. 
Since the balls in $\cB$ is produced by the greedy ball cover procedure, each such optimal cluster has exactly one responsible ball in $\cB$.
Otherwise, since $\cB$ is a ball cover, make any ball  containing the center of $\Pi^*_i$ as the responsible ball (by breaking ties arbitrary). Furthermore, for each ball $B_\ell \in \cB$, let $S_\ell$ denote the indices of optimal clusters for which it is responsible. Note that $\{S_\ell\}_{\ell \in [m]}$ partitions the  indices of the optimal clusters.

Next, let $a_\ell \in S_\ell$ be such that the center $p_\ell$ of $B_\ell$ lies in $\Pi^*_{a_\ell}$. Since $\cB$ is $\Pi^*$-respecting, we have $\Pi^*_{a_\ell} \subseteq B_\ell$.
Consider the choice $(t_1,\dots,t_m)$ such that for each $\ell$:
\begin{itemize}
    \item if $S_\ell \setminus \{a_\ell\} = \emptyset$, then $t_\ell = 0$;
    \item otherwise, $t_\ell \in S_\ell \setminus \{a_\ell\}$ maximizes $r'_i$ over $i \in S_\ell$.
\end{itemize}
Now, consider the corresponding extended balls, $\cF=\{\Ext(B_\ell, r'_{t_\ell})\}_{\ell\in [m]}$. We claim that $\cF$ is a $\Pi^*$-feasible cover.

\paragraph{Correctness.}
Let $\Pi_i^*$ be any cluster and let $B_\ell \in \cB$ be its responsible ball. Then $i \in S_\ell$.

If $i = a_\ell$, then $\Pi_i^* \subseteq B_\ell$. Otherwise, $r'_{t_\ell} \ge r'_i \ge r_i^*$.
Thus, for any $x \in \Pi_i^*$,
\[
d(p_\ell,x) \le d(p_\ell,o_i^*) + d(o_i^*,x) \le \rad(B_\ell) + r_i^* \le \rad(B_\ell)  + r'_{t_\ell},
\]
where $\rad(B_\ell)$ is the radius of $B_\ell$. Thus, $\Pi_i^* \subseteq \Ext(B_\ell, r'_{t_\ell})$.
Therefore, every cluster $\Pi_i^*$ is fully contained in some ball of $\cF$, and hence $\cF$ is a $\Pi^*$-feasible cover.
\noindent
\textbf*{Cost analysis.}
We rely on the crucial observation that the indices used in the expansion step correspond to clusters that are not already accounted for by the original cover $\cB$. Hence, each radius is charged at most once overall. We use this intuition to bound the cost.

Since $\cB$ is produced by the greedy ball cover procedure, each ball $B_\ell=\ball(p_\ell,R_\ell)$ corresponds to the unique cluster index $a_\ell$ in the optimal solution such that $
R_\ell = 2r'_{i_\ell},
$
and the indices $\{a_\ell\}_{\ell=1}^m$ are distinct.
Thus,
\[
\cost(\cB) = \sum_{\ell=1}^m R_\ell = 2 \sum_{\ell=1}^m r'_{a_\ell}.
\]

In the expansion step, we choose $t_\ell \in S_\ell \setminus \{a_\ell\}$, and hence the indices $\{t_\ell\}$ lie in $[k] \setminus \{a_\ell\}_{\ell}$. Therefore,
$
\sum_{\ell=1}^m r'_{t_\ell} \le \sum_{i \notin \{a_\ell\}} r'_i$.
Combining the two bounds,
\[
\cost(\cF) 
= \sum_{\ell=1}^m (R_\ell + r'_{t_\ell})
\le 2 \sum_{i \in \{a_\ell\}} r'_i + \sum_{i \notin \{a_\ell\}} r'_i
\le 2 \sum_{i=1}^k r'_i.
\]

Since $R'$ is an $\eps$-approximation of $R^*$, we obtain $
\cost(\cF) \le (2+\eps)\opt(\cI)$.

\Cref{algo:fc} enumerates over all choices of $(t_1,\dots,t_m)$, where $t_i \in \{0\}\cup [k]$, and therefore there is a branch corresponding to the correct choice of $(t_1,\dots,t_m)$ that constructs $\cF$.
 
 \noindent
\textbf{Running time.}
There are at most $(k+1)^k = 2^{O(k \log k)}$ choices of $(t_1,\dots,t_m)$. Each candidate is constructed in linear time, yielding the claimed running time.
\end{proof}

\subsection{A $(2+\eps)$-Approximation with an Assignment Subroutine}\label{ss:2apx}

We now show how to obtain a $(2+\eps)$-approximation for \msr under mergeable constraints, assuming access to an assignment subroutine (see \Cref{def:assgn}). The algorithm combines the greedy ball cover procedure with the feasible cover construction, and then invokes the assignment subroutine to recover a clustering.

\TwoApx*

\begin{algorithm}[h]
\caption{$(2+\eps)$-Approximation via Assignment Subroutine} \label{algo:2apx}
\KwData{Instance $\cI$ of \msr under a mergeable constraint, parameter $\eps>0$}
\KwResult{A feasible clustering}

Compute a list $\cL$ of radius profiles using \Cref{thm:epsapxrad}\;

$\text{best} \leftarrow \bot$, $\text{bestCost} \leftarrow \infty$\;

\ForEach{$R' \in \cL$}{
    Run GBC with $R'$ to obtain family $\mathbb{B}$\;

    \ForEach{$\cB \in \mathbb{B}$}{
        Run FC on $(\cB, R')$ to obtain family $\mathbb{F}$\;

        \ForEach{$\cF' \in \mathbb{F}$}{
            Run $\cA_A$ on $\cF'$\;
            \If{$\cA_A$ returns a clustering $(C,\sigma)$}{
                \If{$\cost_{\cI}((C,\sigma)) < \text{bestCost}$}{
                    $\text{best} \leftarrow (C,\sigma)$\;
                    $\text{bestCost} \leftarrow \cost_{\cI}((C,\sigma))$\;
                }
            }
        }
    }
}

\Return $\text{best}$;
\end{algorithm}

\begin{proof}
The pseudocode of our algorithm is shown in~\Cref{algo:2apx}.
Let $\Pi^*$ be an optimal clustering with radius profile $R^*$. By \Cref{thm:epsapxrad}, there exists $R' \in \cL$ that is an $\eps$-approximation of $R^*$.

For this choice of $R'$, by \Cref{thm:gbc}, there exists $\cB \in \mathbb{B}$ that is a $\Pi^*$-respecting cover with cost at most $(2+\eps)\opt(\cI)$.

For this $\cB$, by \Cref{thm:feascover}, there exists $\cF \in \mathbb{F}$ that is a $\Pi^*$-feasible cover and satisfies $\cost(\cF) \le (2+\eps)\opt(\cI)$.

Since $\cF$ is $\Pi^*$-feasible, using~\Cref{obs:fc}, there exists a feasible clustering $(C,\sigma)$ with cost at most $\cost(\cF)$.
Hence, the assignment subroutine $\cA_A$ succeeds on input $\cF$ and returns such a clustering. By definition of $\cA_A$, we have
\[
\cost_{\cI}((C,\sigma)) \le \cost(\cF) \le (2+\eps)\opt(\cI).
\]

Since the algorithm returns the minimum-cost clustering among all candidates, the output cost is at most $(2+\eps)\opt(\cI)$.

The running time follows from enumerating all choices of $R' \in \cL$, $\cB \in \mathbb{B}$, and $\cF \in \mathbb{F}$, and invoking $\cA_A$ on each candidate, yielding a total running time of $2^{O(k \log (k/\eps))} \cdot n^2 \cdot T(\cA_A)$.
\end{proof}

\subsection{A Simple $(4+\eps)$-Approximation}\label{ss:4apx}

In this section, we present a simple $(4+\eps)$-approximation algorithm for \msr under mergeable constraints. The algorithm avoids the use of an assignment subroutine and instead exploits the structure of a feasible cover.

\begin{theorem}\label{thm:4apx}
There is an algorithm that given an instance $\cI=\msrins$ of \msr under a mergeable constraint, and any $\eps>0$, computes a feasible clustering of cost at most $(4+\eps)\opt(\cI)$ in time $2^{O(k \log (k/\eps))} \cdot \poly(n)$.
\end{theorem}

\begin{proof}
We first prove the following lemma to compute a clustering for $\cI$ using a $k$-ball cover for $\cI$.
\begin{lemma}\label{lem:4apx}
There is an algorithm that, given a $k$-ball cover $\cF$ for an instance $\cI=\msrins$ of \msr under a mergeable constraint, computes in polynomial time a clustering $(C,\sigma)$ of cost at most $2\cdot\cost(\cF)$ such that if $\cF$ is a feasible cover for $\cI$, then $(C,\sigma)$ is a feasible clustering for $\cI$.
\end{lemma}

\begin{proof}

We define the intersection graph $G$ over $\cF$, where each vertex corresponds to a ball in $\cF$, and two vertices are connected if the corresponding balls intersect. We process each connected component of $G$ independently.
Let $C$ be any connected component of $G$. Denote by $\cF_C$ the set of balls in $C$. The pseudocode is shown in~\Cref{algo:4apx}

\begin{algorithm}[h]
\caption{Component Merging (MC)} \label{algo:4apx}
\KwData{Instance $\cI$ of \msr under a mergeable constraint and a $k$-ball cover  $\cF = \{B_1,\dots,B_m\}$ for $\cI$}
\KwResult{A clustering for $\cI$}

Construct the intersection graph $G$ over $\cF$\;

\ForEach{connected component $C$ of $G$}{
    Let $\cF_C \subseteq \cF$ be the balls in $C$\;
    Let $X_C \leftarrow \bigcup_{B \in \cF_C} B$\;
    Pick an arbitrary point $c_C \in X_C$\;
    Assign all points in $X_C$ to center $c_C$\;
}

\Return the resulting clustering\;
\end{algorithm}

We now bound the radius of each constructed cluster. Fix a component $C$ and pick any ball $B_0 = \ball(c_0, r_0) \in \cF_C$ that contains $c_C$ (the center chosen by the algorithm for $C$). Let $x$ be any point in $X_C = \bigcup_{B \in \cF_C} B$. Since $C$ is connected, there exists a sequence of balls $B_0, B_1, \dots, B_t$ in $\cF_C$ such that each consecutive pair intersects and $x \in B_t$.

Using triangle inequality along this sequence, we obtain
\[
d(c_C, x) \le  2\sum_{j=0}^{t} r_j  \le 2 \sum_{B_i \in \cF_C} r_i.
\]
Thus, the radius of the cluster corresponding to $C$ is at most $2 \sum_{B_i \in \cF_C} r_i$.

Summing over all connected components, the total cost of the clustering is at most
\[
\sum_C 2 \sum_{B_i \in \cF_C} r_i = 2 \sum_{i=1}^m r_i  = 2\cdot \cost(\cF),
\]
using the fact that $\{\cF_C\}_C$ forms a partition of $\cF$. 
\paragraph*{Feasibility.}
Suppose $\cF$ is a $\Pi^*$-feasible cover for some optimal clustering $\Pi^*$ to $\cI$. Then there exists a ball $B_j \in \cF$ such that $\Pi_i^* \subseteq B_j$. Now consider any other ball $B_{j'} \in \cF$ that intersects $\Pi_i^*$. Since $\Pi_i^* \subseteq B_j$, we have $B_{j'} \cap B_j \neq \emptyset$, and hence $B_{j'}$ is adjacent to $B_j$ in $G$. Therefore, all balls that intersect $\Pi_i^*$ lie in the same connected component of $G$. In particular, all points of $\Pi_i^*$ are contained in the union of balls of a single component.
This implies that no optimal cluster is split across multiple components. Hence, merging all points within a component corresponds to merging entire optimal clusters, and by mergeability of the constraint, this yields a feasible cluster.

\end{proof}
We now combine \Cref{algo:4apx} with the procedures developed earlier to obtain the complete $(4+\eps)$-approximation algorithm shown in \Cref{algo:4apxcomplete}.
For analysis, fix an optimal clustering $\Pi^*$ with radius profile $R^*$, and let $R'$ be an $\eps$-approximation of $R^*$. By~\Cref{thm:epsapxrad}, such an $R'$ appears in $\cL$.
Consider this iteration. By~\Cref{thm:gbc}, there exists a cover $\cB \in \mathbb{B}$ that is $\Pi^*$-respecting and satisfies
\[
\cost(\cB) \le (2+\eps)\opt(\cI).
\]
Applying~\Cref{thm:feascover}, there exists $\cF \in \mathbb{F}$ such that $\cF$ is a $\Pi^*$-feasible cover and
\[
\cost(\cF) \le (2+\eps)\opt(\cI).
\]
Since $\cF$ is a feasible cover, we can apply~\Cref{lem:4apx} to obtain a feasible clustering $(C,\sigma)$ with
\[
\cost_{\cI}((C,\sigma)) \le 2\cdot \cost(\cF)
\le 2\cdot(2+\eps)\opt(\cI)
= \left(4 + 2\eps\right)\opt(\cI).
\]
By scaling $\eps$, this yields a $(4+\eps)$-approximation.
Since the algorithm returns the minimum-cost feasible solution encountered, it outputs a clustering with cost at most $(4+\eps)\opt(\cI)$.


\begin{algorithm}[h]
\caption{$(4+\eps)$-Approximation} \label{algo:4apxcomplete}
\KwData{Instance $\cI$  of \msr under a mergeable constraint, parameter $\eps>0$}
\KwResult{A feasible clustering}

Compute a list $\cL$ of radius profiles using \Cref{thm:epsapxrad}\;

$\text{best} \leftarrow \bot$, $\text{bestCost} \leftarrow \infty$\;

\ForEach{$R' \in \cL$}{
    Run GBC with $R'$ to obtain family $\mathbb{B}$\;

    \ForEach{$\cB \in \mathbb{B}$}{
        Run FC on $(\cB, R')$ to obtain family $\mathbb{F}$\;

        \ForEach{$\cF \in \mathbb{F}$}{
            Run MC on $\cI$ and $\cF$ to get clustering $(C,\sigma)$\;
            \If{$(C,\sigma)$ satisfies the mergeable constraint}{
                \If{$\cost_{\cI}((C,\sigma)) < \text{bestCost}$}{
                    $\text{best} \leftarrow (C,\sigma)$\;
                    $\text{bestCost} \leftarrow \cost_{\cI}((C,\sigma))$\;
                }
            }
        }
    }
}

\Return $\text{best}$;
\end{algorithm}

\end{proof}

\subsection{Improvement to factor $(8/3+\eps)$}\label{ss:mainsec}

In this section, we prove our main algorithmic result.


\MainApx*
\begin{proof}
The main ingredient of our algorithm is the following structural result from~\cite{DBLP:conf/approx/BandyapadhyayC25}, which improves upon the clustering cost obtained from~\Cref{lem:4apx} when given a $k$-ball cover. The key idea is that, instead of selecting an arbitrary center from each connected component, one chooses a point that minimizes its maximum distance to any other point in the component, leading to a tighter bound.

\begin{restatable}{lemma}{BCStructural}\label{lem:bc43}
There is a polynomial-time algorithm that, given a $k$-ball cover $\cF$ for an instance $\cI=\msrins$ of \msr under a mergeable constraint, computes a clustering $(S,\sigma)$ of cost at most $\tfrac{4}{3}\cost(\cF)$. Moreover, if $\cF$ is a feasible cover for $\cI$, then $(S,\sigma)$ is a feasible clustering for $\cI$.
\end{restatable}

For completeness, we provide the proof in~\Cref{app:omm}.
We remark that the above bound is tight; see~\Cref{fig:tightbound}.

The pseudocode of our algorithm is given in~\Cref{algo:83apx}. It follows the same structure as~\Cref{algo:4apxcomplete}, except that we invoke the algorithm of~\Cref{lem:bc43} on each candidate cover.

\begin{algorithm}[h]
\caption{$(\tfrac{8}{3}+\eps)$-Approximation} \label{algo:83apx}
\KwData{Instance $\cI$  of \msr under a mergeable constraint, parameter $\eps>0$}
\KwResult{A feasible clustering}

Compute a list $\cL$ of radius profiles using \Cref{thm:epsapxrad}\;

$\text{best} \leftarrow \bot$, $\text{bestCost} \leftarrow \infty$\;

\ForEach{$R' \in \cL$}{
    Run GBC with $R'$ to obtain family $\mathbb{B}$\;

    \ForEach{$\cB \in \mathbb{B}$}{
        Run FC on $(\cB, R')$ to obtain family $\mathbb{F}$\;

        \ForEach{$\cF \in \mathbb{F}$}{
            Use the algorithm of~\Cref{lem:bc43} on $\cF$ to obtain clustering $(C,\sigma)$\;
            \If{$(C,\sigma)$ satisfies the mergeable constraint}{
                \If{$\cost_{\cI}((C,\sigma)) < \text{bestCost}$}{
                    $\text{best} \leftarrow (C,\sigma)$\;
                    $\text{bestCost} \leftarrow \cost_{\cI}((C,\sigma))$\;
                }
            }
        }
    }
}

\Return $\text{best}$\;
\end{algorithm}

\paragraph*{Analysis.}
Fix an optimal clustering $\Pi^*$ with radius profile $R^*$, and let $R'$ be an $\eps$-approximation of $R^*$. By~\Cref{thm:epsapxrad}, such an $R'$ appears in $\cL$.

Consider this iteration. By~\Cref{thm:gbc}, there exists a cover $\cB \in \mathbb{B}$ that is $\Pi^*$-respecting and satisfies
\[
\cost(\cB) \le (2+\eps)\opt(\cI).
\]
Applying~\Cref{thm:feascover}, there exists $\cF \in \mathbb{F}$ such that $\cF$ is a $\Pi^*$-feasible cover and
\[
\cost(\cF) \le (2+\eps)\opt(\cI).
\]
Since $\cF$ is a feasible cover, we can apply~\Cref{lem:bc43} to obtain a feasible clustering $(C,\sigma)$ with
\[
\cost_{\cI}((C,\sigma)) \le \tfrac{4}{3} \cost(\cF)
\le \tfrac{4}{3}(2+\eps)\opt(\cI)
= \left(\tfrac{8}{3} + \tfrac{4}{3}\eps\right)\opt(\cI).
\]
By scaling $\eps$, this yields a $(\tfrac{8}{3}+\eps)$-approximation.
Since the algorithm returns the minimum-cost feasible solution encountered, it outputs a clustering with cost at most $(\tfrac{8}{3}+\eps)\opt(\cI)$.


\noindent
\textbf{Running time.}
The list $\cL$ has size $2^{O(k \log (k/\eps))} \cdot n^2$ and can be computed in the same time by~\Cref{thm:epsapxrad}. For each $R'$, GBC produces at most $k^k = 2^{O(k \log k)}$ candidates, and FC produces at most $(k+1)^k = 2^{O(k \log k)}$ candidates per cover. The procedure of~\Cref{lem:bc43} runs in polynomial time.
Thus, the overall running time is $2^{O(k \log (k/\eps))} \cdot \poly(n)$,
as claimed.
\begin{figure} \centering \includegraphics[width=0.35\linewidth]{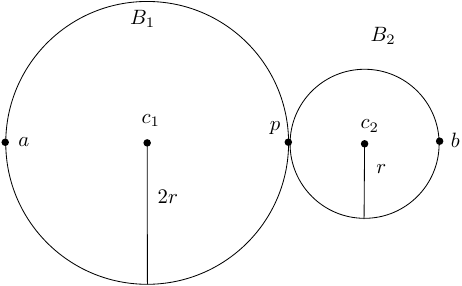} \caption{\small An example of a $2$-cover in which  the minimum cost is $4/3$ of the sum of radii of the balls.} \label{fig:tightbound} \end{figure}
\end{proof}

\subsection{Assignment Subroutines}\label{ss:assgn}
In this section, we present assignment subroutines for several clustering problems. These subroutines, together with~\Cref{thm:2apx}, imply a $(2+\eps)$ \fpt-approximation for all these problems.

\paragraph*{\abfair.}
The assignment subroutine for \abfair follows from~\cite{DBLP:journals/jcss/BandyapadhyayFS24}, who designed an assignment subroutine for \abfairsc under the $k$-median and $k$-means objectives. The key idea in their work is to solve a mixed-integer linear program with $O(k)$ many integer variables. This allows us to find an optimal assignment in \fpt\ time using the celebrated result of Lenstra~\cite{DBLP:journals/mor/Lenstra83} and subsequent improvements by~\cite{DBLP:journals/mor/Kannan87,DBLP:journals/combinatorica/FrankT87}.

Recall that in \abfair, there are $\ell$ groups $X_1,\dots,X_\ell$ over the point set $X$, where $\ell$ is a constant. These groups may overlap, and we refer to them as colors. For every color $c \in [\ell]$, we are given a range $[\alpha_c,\beta_c]$, specifying the fraction of points of color $c$ that must be present in any cluster. Let $\Gamma$ denote the number of distinct collections of groups that a point can belong to; we refer to these collections as equivalence classes, or simply classes. 
For a class $\tau \in [\Gamma]$, let $X^\tau$ denote the set of points belonging to class $\tau$. 

In the assignment problem, we are given a set of balls $\cB=\{B_1,\dots, B_k\}$, and we want to compute an $(\alpha,\beta)$-fair clustering by assigning each point to the center of a ball that contains it. For a point $p \in X$, let $\gamma(p)$ denote the set of balls in $\cB$ that contain $p$.

\begin{align}
\text{Variables:} \quad 
& f_{p,i} \in [0,1] && \forall p \in X,\; i \in \gamma(p) \\
& L_{\tau,i} \in \mathbb{Z}_{\ge 0} && \forall \tau\in [\Gamma],\; i \in [k]\\[6pt]
\text{(Assignment)} \quad 
& \sum_{i \in \gamma(p)} f_{p,i} = 1 
&& \forall p \in X \\[6pt]
\text{(Class loads)} \quad 
& L_{\tau,i} = \sum_{p \in X^\tau \cap B_i} f_{p,i}
&& \forall \tau \in [\Gamma],\; i \in [k] \\[6pt]
\text{(Lower bound constraints)} \quad 
& \sum_{p \in X_c \cap B_i} f_{p,i}
\;\ge\;
\alpha_c \sum_{p \in X \cap B_i} f_{p,i}
&& \forall i \in [k],\; c \in [\ell] \\[6pt]
\text{(Upper bound constraints)} \quad 
& \sum_{p \in X_c \cap B_i} f_{p,i}
\;\le\;
\beta_c \sum_{p \in X \cap B_i} f_{p,i}
&& \forall i \in [k],\; c \in [\ell]
\end{align}

Since we only seek a feasible assignment, the objective is a dummy objective:
\begin{align*}
    \text{minimize} \quad  0.
\end{align*}

In this program, we require only the variables $L_{\tau,i}$—which correspond to the total number of points of class $\tau \in [\Gamma]$ assigned to ball $i \in [k]$—to be integral. It is known that this program can be solved optimally in time $(k\Gamma)^{O(k\Gamma)} \poly(n)$, which is \fpt\ in $k$, since $\Gamma$ is a constant. Using this solution,~\cite{DBLP:journals/jcss/BandyapadhyayFS24} show how to obtain an integral clustering (i.e., an assignment of each point to a single ball) in polynomial time via a flow-based rounding procedure.

\medskip
\noindent
\textbf{Balanced clustering and clustering with lower bounds.}
For both these problems, \cite{DBLP:conf/isaac/CartaDHR024} provide a simple flow-based polynomial-time assignment subroutine.


\section{Hardness of \msr}\label{sec:hard}

In this section, we establish  parameterized hardness for \msr. In particular, we show that the classical (unconstrained) version of \msr is \Wtwo-hard when parameterized by $k$, and moreover, it does not admit an efficient parameterized approximation scheme (EPAS).

\MSRHardness*
\begin{proof}
We give a parameterized reduction from the following variant of the Set Cover problem.

\paragraph*{\scover.}
We are given a universe $\cU$ and $k$ collections of subsets $\cC_1,\dots,\cC_k \subseteq 2^{\cU}$. The task is to select one set from each collection such that their union covers $\cU$. The parameter is $k$. This problem is known to be \Wtwo-hard.

Given an instance $\cJ$ of \scover, we construct an instance $\cI$ of \msr such that, if $\cJ$ admits a set cover of size $k$, then the optimal cost of $\cI$ is at most $2^k-1$; otherwise, every feasible solution to $\cI$ has cost at least $2^k$.

\paragraph*{Construction.}
We define a graph $G$ whose vertex set consists of:
\begin{itemize}
    \item (element vertices) a vertex for each element $u \in \cU$,
    \item (set vertices) a vertex for each set $S \in \bigcup_{i=1}^k \cC_i$,
    \item (auxiliary vertices) for each collection $\cC_i$, a set of $k+1$ auxiliary vertices $A_i$
\end{itemize}

We add edges as follows:
\begin{itemize}
    \item For each set $S$ and each element $u \in S$, add an edge between $S$ and $u$.
    \item For each collection $\cC_i$, connect all sets in $\cC_i$ into a clique.
    \item For each collection $\cC_i$, connect every auxiliary vertex in $A_i$ to every set in $\cC_i$.
\end{itemize}

We assign edge weights as follows:
\begin{itemize}
    \item All edges incident to sets in $\cC_i$ have weight $2^{i-1}$.
\end{itemize}

The metric $d$ is defined as the shortest-path metric induced by this weighted graph. The parameter for the \msr instance is $k$. Let $\cI$ be the resulting instance of \msr.

\paragraph*{Completeness.}
Suppose there exists a solution to the \scover instance $\cJ$, i.e., there are sets $S_1 \in \cC_1,\dots,S_k \in \cC_k$ whose union covers $\cU$.
Our solution to $\cI$ chooses the centers to be these corresponding $k$ sets. 
Let $C=\{c_1,\dots,c_k\}$ be the set of obtained centers such that $c_i$ corresponds to set $S_i$.
We assign each element $u \in \cU$ to any selected set that contains it, and, for $i \in [k]$, we assign all vertices in $A_i \cup \cC_i$  to $c_i \in C$. Let $(C,\sigma)$ be the resulting solution.

We now bound the cost of $(C,\sigma)$. For each $i \in [k]$, all vertices in $\cC_i \cup A_i$ are within distance at most $2^{i-1}$ from $c_i$, and any element assigned to $c_i$ is also at distance at most $2^{i-1}$. Thus, $\rad(c_i) \le 2^{i-1}$. Therefore, we have that $\opt(\cI) \le \cost((C,\sigma)) =
\sum_{i=1}^k 2^{i-1} = 2^k - 1$.

\paragraph*{Soundness.}
The soundness of the construction crucially relies on the following claim, whose proof is similar to one present in~\cite{GibsonKKPV12}.
\begin{claim}\label{cl:soundness}
    For any solution $(C,\sigma)$ to $\cI$, if  $\cost((C,\sigma)) \le 2^k-1$, then, for every $i \in [k]$, $C$ contains a vertex from $\cC_i$ with radius at least $2^{i-1}$.
\end{claim}
\begin{proof}
    For contradiction, assume there is a solution $(C,\sigma)$ with cost at most $2^k-1$ but there is no center from some $\cC_i$ with radius at least $2^{i-1}$. Let $\ell \in [k]$ be the largest index such that there is no center in $C$ from $\cC_\ell$ with radius at least $2^{\ell-1}$. As there are $k+1$ vertices in $A_\ell$, there is some $a \in A_\ell$ that is not a center in $C$. Now, let $c \in C$ be the center such that $\sigma(a)=c$. Then, first note that $c \notin\cC_i$, since $d(c,a) = 2^{\ell-1}$, while $\rad(c) < 2^{\ell-1}$. Since the distance of $a$ from any vertex other than vertices in $\cC_\ell$ is at least $2\cdot 2^{\ell-1}=2^\ell$,    we have that $d(a,c) \ge 2^\ell$. Therefore, we have that $\cost((C,\sigma)) \ge 2^\ell + \sum_{i = \ell+1}^k 2^{i-1} = 2^k$, a contradiction to the cost of $(C,\sigma)$.
\end{proof}

Suppose there is no  solution to the \scover instance $\cJ$. We show that any clustering using at most $k$ centers has cost at least $2^k$. For contradiction, assume that there is a solution $(C,\sigma)$ to $\cI$ with $k$ centers and $\cost((C,\sigma)) < 2^k$. Since the edge weights are integers, we have that $\cost((C,\sigma)) \le 2^k-1$. Then, using~\Cref{cl:soundness}, we have that, for every $i \in [k]$,  there exists $c_i \in C$ such that $c_i \in \cC_i$. Pick the sets corresponding to the centers in $C$, and let $\cS=\{S_1,\dots, S_k\}$ be the resulting collection, such that $S_i$ is the set corresponding to $c_i$. Then, we have that,  for every $i \in [k]$, it holds that $|\cS \cap \cC_i|=1$. Now, we claim that the sets in $\cS$ cover $\cU$, contradicting the fact that there is no solution to $\cJ$. To see this, using~\Cref{cl:soundness} and the fact that $\cost((C,\sigma)) \le 2^k-1$, we have that $\rad(c_i) = 2^{i-1}$. Therefore, every assigned element vertex to $c_i$ is contained in $S_i$, and hence $\cU \subseteq \bigcup_{i \in [k]} S_i$.

This finishes the proof that \msr is \Wtwo-hard.

\paragraph*{Gap and inapproximability.}
The reduction yields:
\begin{itemize}
    \item If $\cJ$ is a YES-instance, then $\opt(\cI) \le 2^k - 1$.
    \item If $\cJ$ a NO-instance, then $\opt(\cI) \ge 2^k$.
\end{itemize}

Therefore, there is a gap of at least a factor $\frac{2^k}{2^k - 1} = 1 + \frac{1}{2^k - 1}$. 

Now, suppose that there is an algorithm $\cA$ that given an instance of \msr and any $\eps>0$, outputs a $(1+\eps)$-approximate solution in running time $f(k,\eps) \poly(n)$. Then, we show that \scover admits an \fpt-algorithm. Given an instance $\cJ$ of \scover, we first reduce it to obtain an instance $\cI$ of \msr using the above reduction. Then, we use $\cA$ on $\cI$ with $\eps=2^{-k}$, and output YES if the cost of solution returned by $\cA$ is $< 2^k$, and NO otherwise. For correctness, suppose $\cJ$ is an YES instance, then we have $\opt(\cI) \le 2^k-1$, and hence $\cA$ returns a solution with cost $\le (1+\eps) (2^k-1) = (2^k - 2^{-k}) < 2^k$, and hence our algorithm outputs YES, as desired. Now, suppose $\cJ$ is a NO instance, and hence $\opt(\cI) \ge 2^k$. Therefore, our algorithm outputs NO, as desired.

For the \eth-hardness, the same argument goes through with the fact that there is no $g(k,\eps)|\cJ|^{o(k)}$ algorithm for \scover, assuming \eth.

\end{proof}

\subsection{Limitations of Previous Hardness Result}\label{ss:hardissue}

We briefly explain the issue with the lower-bound construction of Bandyapadhyay and Chen~\cite{DBLP:conf/approx/BandyapadhyayC25}. Their reduction is from \domset. Given a graph $G$, the construction augments $G$ by adding two sets of $k$ vertices and connecting each such vertex to every vertex of $G$. We refer to these as \emph{auxiliary vertex sets}, and to their elements as \emph{auxiliary vertices}. 
The groups are defined as follows: one group for each auxiliary vertex set and one group corresponding to the original vertex set of $G$. The proportions are chosen so that, for each auxiliary group, $\alpha_i = 1/(n+3)$ and $\beta_i = 1$, while for the original group, $\alpha = 1/3$ and $\beta = 1$, where $n$ denotes the total number of vertices in the constructed \msr instance. The metric is given by the shortest-path distances in the resulting graph.

In the yes-instance, there is a feasible \msr solution of cost $k$, obtained by selecting the vertices corresponding to a dominating set of $G$. For the no-instance, the argument in~\cite{DBLP:conf/approx/BandyapadhyayC25} relies on the observation that singleton clusters violate the fairness constraints, and therefore concludes that zero-radius clusters cannot occur. This is then used to argue that any feasible solution to \msr can be used to construct a dominating set of size $k$.

The issue is that, while singleton clusters are indeed infeasible, this does not preclude the existence of zero-radius clusters. In particular, empty clusters satisfy the fairness constraints and have zero radius. Exploiting this, one can construct a low-cost solution as follows: choose a center from one of the auxiliary sets, which—via the added edges—covers all points within distance $2$. This forms a large cluster that satisfies the fairness constraints. The remaining $k-1$ clusters can then be opened at arbitrary vertices but left empty (their centers are assigned to the large cluster). The resulting solution remains feasible and has total cost at most $2$, contradicting the intended hardness gap.

More generally, this highlights a limitation of reductions that introduce auxiliary vertices enabling \emph{shortcuts} in the metric. Under mergeable constraints, such constructions can collapse into a single large cluster, while the remaining clusters can be left empty without violating feasibility, thereby yielding artificially low-cost solutions.
\bibliography{refs}
\bibliographystyle{plain}

\appendix

\section{Ommitted Proof}\label{app:omm}
\BCStructural*
\begin{proof}
Similar to the proof in~\Cref{lem:4apx}, we define the intersection graph $G$ over $\cF$, where each vertex corresponds to a ball in $\cF$, and two vertices are connected if the corresponding balls intersect. We process each connected component of $G$ independently.
Let $C$ be any connected component of $G$. Denote by $\cF_C$ the set of balls in $C$, and let $\cR_C$ be the sum of radii of balls in $\cF_C$. 
Let $X_C$ be the union of the points contained in the balls of $\cF_C$.

\begin{claim}\label{cl:bc43}
    There exists a point $p_C \in X_C$ such that $d(p_C,q) \le \tfrac{4}{3}\cR_C$, for any $q \in X_C$.
\end{claim}
\begin{proof}
    Suppose there is a ball $B=\ball(c,R)$ in $\cF_C$ such that $R \geq \tfrac{2}{3}\cR_C$, then consider the center $c$ of $B$. We have that, using triangle inequality, for any $q \in X_C$,
    \[
    d(c,q) \leq 2\cR_C - R \le \tfrac{4}{3}\cR_C,
    \]
    since there is a path from $B$ to a ball containing $q$ in $C$. Thus, the claim is true.

    Now, suppose 
\begin{align}\label{eqn:basic}
    \text{every ball in } \cF_C \text{ has radius }< \tfrac{2}{3}\cR_C.
\end{align} 
    Let $T_C$ be a spanning tree of $C$, and let $\mu=(\mu_1,\dots,\mu_\ell), \ell \le k$ be a leaf to leaf path such that the sum of radii of the balls along $\mu$ is maximum among all leaf to leaf paths. Observe that removing the vertices of $\mu$ from $T_C$ disconnects $T_C$ into connected components. Moreover, each such connected component of $T_C$ is connected to exactly one $\mu_i \in \mu$ via exactly one vertex of the  component. We say that such component is connected to $T_C$ via $\mu_i$.
    
    For each $\mu_i \in \mu$, we denote by $\ball(\mu_i)$ as the ball in $\cF_C$ corresponding to the vertex $\mu_i$ with radius $R_i$. Furthermore, denote by $\cB_{\mu_i}$ the set of balls of all components of $T_C - \mu$ that are connected via $\mu_i$. Moreover, $\cB_{\leq  \mu_i}:=\bigcup_{j\leq i} \cB_{\mu_j}$, and $\cB_{\geq \mu_i}:=\bigcup_{j\geq i} \cB_{\mu_j}$.
    Also, denote by $\ballp_{\mu_i} = \{\ball(\mu_i)\} \cup \cB_{\mu_i}$, and $\cost(\ballp_{\mu_i})$ as the sum of radii of balls in $\ballp_{\mu_i}$. Finally, we have the following notations for every $\mu_i \in \mu$.
    \begin{itemize}
        \item $\lsum(\mu_i):=  \sum_{j<i} \cost(\ballp_{\mu_j})$
       \item $\rsum(\mu_i):=  \sum_{j>i} \cost(\ballp_{\mu_j})$
      \item $\msum(\mu_i):=  \cost(\ballp_{\mu_i})$
     \item $\lmsum(\mu_i):= \lsum(\mu_i) +\msum(\mu_i)$
    \end{itemize}

Note that $\lsum(\mu_i)+ \msum(\mu_i) +\rsum(\mu_i) = \cost(\cF_C) = \cR_C$.

\begin{definition}\label{def:critical}
    Let $\mu_{i^*} \in \mu$ be such that $i^* \in [\ell]$ is the smallest index such that $\lmsum(\mu_{i^*}) \ge \tfrac{2}{3}\cR_C$.
\end{definition}
 First, note that $i^* >1$, since $\mu_1$ is a leaf in $T_C$ and the radius of $\ball(\mu_1) < \tfrac{2}{3}\cR_C$. Next, we divide into two cases.

\paragraph{Case 1. $\lsum(\mu_{i^*}) \ge \frac{\cR_c}{3}$.}
Since $i^*>1$, let $p_C$ be the point in $\ball(\mu_{i^*-1}) \cap \ball(\mu_{i^*})$. Then, for any $q \in \cB_{\geq \mu_{i^*}}$, using triangle inequality, we have
\[
d(p_C, q) \le 2(R_{i^*} + \rsum(\mu_{i^*})) \le \tfrac{4}{3}\cR_C,
\]
since $R_{i^*} + \rsum(\mu_{i^*}) \le \tfrac{2}{3}\cR_C$ in this case.

Now suppose $q \in \cB_{\le \mu_{i^*-1}}$, then by the property of $i^*$, we have that $\lsum(\mu_{i^*-1}) < \tfrac{2}{3}\cR_C$. Hence, 
\[
d(p_C,q) \le 2 \cdot \lmsum(\mu_{i^*-1})  \le \tfrac{4}{3}\cR_C,
\]
as desired.

\paragraph{Case 2. $\lsum(\mu_{i^*}) < \frac{\cR_c}{3}$.}
Let $p_C$ be the center of $\ball(\mu_{i^*})$. We claim that $p_C$ satisfies the desired property. 

\noindent
Towards this, we partition balls in $\cF_C$ into three parts:  $\cF_C^L$ that is union of the balls $\ballp_{\mu_1},\dots, \ballp_{\mu_{i^*-1}}$, $\cF_C^M := \ballp_{\mu_{i^*}}$, and $\cF_c^R$ is the union of the balls  $\ballp_{\mu_{i^*+1}},\dots, \ballp_{\mu_{\ell}}$. Similarly, we partition the points of $X_C$ into three parts: $X_C^L, X_C^M,X_C^R$ containing points from $\cF_C^L, \cF_C^M, \cF_C^R$, respectively. Now consider $q \in X_C$.

\textbf{If $q \in X_C^L$. }
In this case, we have
\[
d(p_C,q) \le R_{i^*} + 2\cdot \lsum(\mu_{i^*}) \le \tfrac{4}{3}\cR_C,
\]
since both $R_{i^*}$ and $2\cdot \lsum(\mu_{i^*})$ are at most $\le \tfrac{2}{3}\cR_C$ due to~\Cref{eqn:basic} and Case 2, respectively.

\medskip
\textbf{If $q \in X_C^R$. }
\[
d(p_C,q) \le R_{i^*} + 2\cdot \rsum(\mu_{i^*}) \le \tfrac{4}{3}\cR_C,
\]
since both $R_{i^*}$ and $2\cdot \rsum(\mu_{i^*})$ are at most $\le \tfrac{2}{3}\cR_C$ due to~\Cref{eqn:basic} and $\lmsum(\mu_{i^*}) \geq \tfrac{2}{3}\cR_C$ (\Cref{def:critical}), respectively.

\medskip
\textbf{If $q \in X_C^M$. }
Let $\cost(\mu)$ denote the sum of radii of balls corresponding to the vertices in $\mu$.
Consider any path $\mu'$ from $\mu_{i^*}$ to a leaf (ball) in $\cF_C^M$ that passes through a ball of $q$. Let $\mu' -\mu_{i^*}$ denote the path obtained from $\mu'$ by removing $\mu_{i^*}$. Then, we have
\begin{align}\label{eqn:eq2}
    \cost(\mu' -\mu_{i^*}) \le \lsum(\mu_{i^*}) + \rsum(\mu_{i^*}),
\end{align}
as otherwise, $\cost(\mu') = \cost(\mu' -\mu_{i^*}) + R_{i^*} > \lsum(\mu_{i^*}) + \rsum(\mu_{i^*})+ R_{i^*} = \cost(\mu)$, contradicting the fact that $\mu$ had maximum cost among all such paths.

On the other hand, since $\cR_C$ is the sum of radii of the balls in $\cF_C$, we have,
\[
\lsum(\mu_{i^*}) + \rsum(\mu_{i^*}) +  \cost(\mu' -\mu_{i^*}) + R_{i^*} \le \cR_C
\]
Hence, using the lower bound of~\Cref{eqn:eq2} with the above inequality, we have that
\begin{align*}
    2\cost(\mu' -\mu_{i^*}) + R_{i^*} \le \cR_C
\end{align*}
Therefore, we have that
\[
d(p_C,q) \le  R_{i^*} +     2\cost(\mu' -\mu_{i^*}) \le \cR_C,
\]
 as desired.   
    \end{proof}
Now we use~\Cref{cl:bc43} to finish the proof of the lemma as follows. For every component $C$, pick the center $s_C$ as the point that minimizes its maximum distance to any other point in $X_C$, and assign all the points in $X_C$ to $s_C$. First note that due to the choice of $s_C$, we have that, $\max_{q \in X_C} d(s_C,q) \le \max_{q \in X_C} d(p_C,q) \le \tfrac{4}{3}\cR_C$, due to~\Cref{cl:bc43}. Let $(S,\sigma)$ denote the resulting solution.
Moreover, since the balls across the components are disjoint, we have that 
   \[
    \cost((S,\sigma)) = \sum_{\text{component } C} \rad(c_C) \le \sum_{\text{component } C} \tfrac{4}{3}\cR_C = \tfrac{4}{3}\cost(\cF).
   \] 

Now, suppose $\cF$ is a $\Pi^*$-feasible cover for some optimal clustering $\Pi^*$ to $\cI$. Then there exists a ball $B_j \in \cF$ such that $\Pi_i^* \subseteq B_j$. Now consider any other ball $B_{j'} \in \cF$ that intersects $\Pi_i^*$. Since $\Pi_i^* \subseteq B_j$, we have $B_{j'} \cap B_j \neq \emptyset$, and hence $B_{j'}$ is adjacent to $B_j$ in $G$. Therefore, all balls that intersect $\Pi_i^*$ lie in the same connected component of $G$. In particular, all points of $\Pi_i^*$ are contained in the union of balls of a single component.
This implies that no optimal cluster is split across multiple components. Hence, merging all points within a component corresponds to merging entire optimal clusters, and by mergeability of the constraint, this yields a feasible cluster.
   
\end{proof}
\end{document}